\documentclass[
  sigconf,
  screen,
]{acmart}

\usepackage[utf8]{inputenc}

\usepackage{microtype}
\usepackage{graphicx}
\usepackage{epstopdf}
\graphicspath{{figures/}{sections/}}
\DeclareGraphicsExtensions{.pdf,.eps,.jpeg,.png}
\usepackage{caption}
\usepackage{makecell}
\usepackage{float}
\usepackage{tabularx}
\usepackage{booktabs}
\usepackage{subcaption}
\usepackage{array}
\usepackage{multirow}
\newcolumntype{L}{>{$}l<{$}} % math-mode version of "l" column type
  \newcolumntype{R}{>{$}r<{$}} % math-mode version of "r" column type
\newcolumntype{C}{>{$}c<{$}} % math-mode version of "r" column type
\newcolumntype{P}[1]{>{\centering\arraybackslash}p{#1}}
\newcolumntype{Y}{>{\centering\arraybackslash}X}
\usepackage[inline]{enumitem}
\usepackage{xcolor}
\AtEndPreamble{
  \usepackage{cleveref}
}

\usepackage[
  ruled,
  lined,
  linesnumbered,
  commentsnumbered,
  longend
]{algorithm2e}
\SetAlgorithmName{Algorithm}{Algorithm}{List of Algorithms}

\usepackage{notation}

\theoremstyle{remark}

\newtheorem*{remark}{Remark}
\newtheorem*{note}{Note}

\usepackage{ifdraft}
\usepackage{comment}

\ifoptionfinal{%
  \excludecomment{todocomment}
  
}{%
  
}

\usepackage[
  type={CC},
  modifier={by},
  version={4.0},
]{doclicense}

\copyrightyear{2025}
\acmYear{2025}
\setcopyright{cc}
\setcctype{by}
\acmConference[HSCC '25]{28th ACM International Conference on Hybrid Systems: Computation and Control}{May 6--9, 2025}{Irvine, CA, USA}
\acmBooktitle{28th ACM International Conference on Hybrid Systems: Computation and Control (HSCC '25), May 6--9, 2025, Irvine, CA, USA}
\acmDOI{10.1145/3716863.3718033}
\acmISBN{979-8-4007-1504-4/2025/05}

\begin{document}

\title{Monitoring Spatially Distributed Cyber-Physical Systems with Alternating Finite Automata}
\author{Anand Balakrishnan}
\email{anandbal@usc.edu}
\affiliation{
  \institution{University of Southern California}
  \city{Los Angeles}
  \country{USA}
}
\author{Sheryl Paul}
\email{sherylpa@usc.edu}
\affiliation{
  \institution{University of Southern California}
  \city{Los Angeles}
  \country{USA}
}
\author{Simone Silvetti}
\email{simone.silvetti@dia.units.it}
\affiliation{
  \institution{University of Trieste}
  \country{Italy}
}
\author{Laura Nenzi}
\email{lnenzi@units.it}
\affiliation{
  \institution{University of Trieste}
  \country{Italy}
}
\author{Jyotirmoy V.~Deshmukh}
\email{jdeshmuk@usc.edu}
\affiliation{
  \institution{University of Southern California}
  \city{Los Angeles}
  \country{USA}
}

\begin{abstract}
  Modern cyber-physical systems (CPS) can consist of various networked
  components and agents interacting and communicating with each other.
  In the context of spatially distributed CPS, these connections can be
  dynamically dependent on the spatial configuration of the various components
  and agents.
  In these settings, robust monitoring of the distributed components is vital
  to ensuring complex behaviors are achieved, and safety properties are
  maintained.
  To this end, we look at defining the automaton semantics for the
  Spatio-Temporal Reach and Escape Logic (STREL), a formal logic designed to
  express and monitor spatio-temporal requirements over mobile, spatially
  distributed CPS.
  Specifically, STREL reasons about spatio-temporal behavior over dynamic
  weighted graphs.
  While STREL is endowed with well defined qualitative and quantitative
  semantics, in this paper, we propose a novel construction of (weighted)
  alternating finite automata from STREL specifications that efficiently
  encodes these semantics.
  Moreover, we demonstrate how this automaton semantics can be used to perform
  both, offline and online monitoring for STREL specifications using
  a simulated drone swarm environment.
\end{abstract}

\maketitle

% !TEX root = ../main.tex

\section{Introduction}

Multi-agent and distributed cyber-physical systems (CPS) are becoming
increasingly pervasive in modern society, from IoT networks \cite{iotcps1, iotcps2,iotcps3}, vehicular
networks \cite{wang2019survey}, power grids
\cite{zhu2020cost,pagani2013power} smart cities \cite{smartcities1,smartcities2,smartcities3}, drone swarms, to
multi-robot configurations used in space exploration \cite{nasa_cadre}.
Individual
nodes (or agents) in such systems communicate and interact with each other, and
the temporal behavior of one node can affect the behavior of other nodes.
A
common approach to model the behavior of such systems is as a time-varying
connectivity graph where individual nodes are considered as stochastic processes
evolving in time.
As such systems become networked at scale, it has become
increasingly important to monitor them to ensure that they perform tasks as
expected while satisfying safety constraints.
For instance, consider the problem
of detecting a real spike (vs.
bad sensor data) for a given bus in a power grid.
A proposed solution in \cite{zhu2020cost} is to compare if some subset of
neighboring busses also display the same spike.
Monitoring such properties
requires reasoning about both the spatial connectivity and temporal behavior of
nodes in the underlying graph representation of the system.

% \textcolor{red}{MORE TEXT NEEDS TO GO HERE} \textcolor{green}{Done- check if ok, wrote motivation for the fig too}
In systems where agents must coordinate tasks or avoid certain areas, spatial
relationships play a critical role, influencing both safety and performance.
For instance, in multi-robot systems, agents may need to stay within a specific
proximity to maintain communication or avoid obstacles in dynamic environments.
Spatio-temporal logics allow designers to specify these constraints and monitor
systems to ensure they conform to such requirements.
For e.g. in ~\ref{fig:example_map3}, we have a simulation in the SwarmLab
environment \cite{soria2020swarmlab} showing a flock of drones that must
navigate through an obstacle-dense environment to reach a central goal.
They must avoid obstacles (indicated by the black polygons), while maintaining
contact with the ground-stations (indicated by the blue regions).
Monitoring such systems requires advanced logic to ensure that agents maintain
connectivity, avoid collisions, and achieve goals efficiently.
\begin{figure}
  \includegraphics[width=0.75\columnwidth]{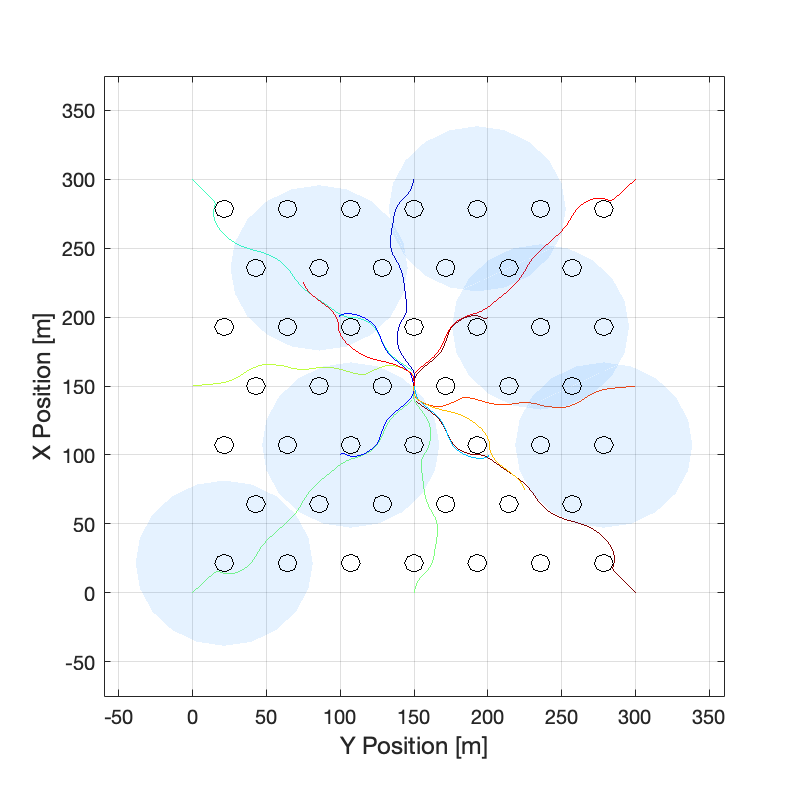}
  \caption{%
    A spatially distributed system of drones coordinating with each other to
    ``flock'' towards a common goal.
    The drones can have requirements to maintain communications with ground stations (blue circles) either directly or through a flock member.
  }\label{fig:example_map3}
\end{figure}
% \textcolor{red}{FIX THIS: 
% To tackle these situations, various spatio-temporal logics have been proposed
% that explicitly reason about the spatial configuration of the system~\cite{Haghighi15_spatel, Nenzi18_SSTL, Ma21, nenzi2022logic,bartocci2017monitoring}.}

In the CPS context, temporal logics have become a popular tool to describe and
monitor for complex spatio-temporal behavior \cite{cpsbook}.
Logics like Metric
Temporal Logic~\cite{mtl1, mtl2, mtl3} and Signal Temporal Logic~\cite{stl1,stl2} (and their
extensions) have been used extensively in this context due to the various
quantitative semantics proposed for them \cite{qtstl1}, in
addition to their qualitative or Boolean semantics.
Recently, there has been a
lot of interest in spatio-temporal logics to explicitly reason about the spatial
configuration of a system.
For example, Signal spatial-temporal logic (SSTL)
\cite{nenzi2015specifying, nenzi2015qualitative} is a modal logic that allows
combining temporal modalities with two spatial modalities, one expressing that
something is true somewhere nearby and the other capturing the notion of being
surrounded by a region that satisfies a given spatio-temporal property.
The
logic SaSTL \cite{sastl} extends STL by introducing operators that allow
spatial aggregation of real-valued quantities at neighbors and spatial counting
across sets of agents.
The logic SpaTel introduces directional spatial operators
by combining STL with the {\it Tree Spatial Superposition
    Logic}~\cite{bartocci2014} (TSSL) and assumes a fixed topology over the space
(for example a checkerboard pattern) \cite{Haghighi15_spatel}.
Most of the
aforementioned logics assume static spatial configurations.

% \textcolor{blue}{Anand? WRITE ABOUT SPATIAL.}

Of particular relevance to this paper is the logic STREL (Spatio-Temporal Reach
and Escape Logic) \cite{nenzi2022logic}.
This logic combines STL with two
spatial operators: reach and escape which allow reasoning about temporal
properties over arbitrary routes in the graph describing the spatial
configuration.
The interesting aspect of STREL is that it allows reasoning about
dynamic spatial configurations, and has been used to model many different kinds
of properties such as bike-sharing systems in smart cities, mobile ad hoc
networks \cite{nenzi2022logic}, and disease spread during a pandemic
\cite{mohammadinejad2021mining}.

In the automata theory literature, the connection between logics and automata is
well-studied.
For example, for every Linear Temporal Logic formula, an
equivalent non-deterministic B\"{u}chi automaton can be constructed in size
exponential in the length of the formula \cite{buchitl}, and can be efficiently constructed \cite{pnueli2008merits}.
Similar constructions exist
for translating computation tree logic (CTL) to weak B\"{u}chi tree automata
\cite{vardi1996automatatheoretic}, to map metric temporal logic to timed
automata \cite{alurstl,maler2006mitl} among others.
Curiously, this direction remains
unexplored for spatio-temporal logics.
In this paper, we present the first such
translation from a discrete-time variant of the logic STREL to an {\em
    alternating finite-state automaton} (AFA).

In a nondeterministic finite automaton (NFA), the notion of acceptance of a word
$w$ is that there exists a path from the initial state that reaches a final
state.
Consider a state $q$ in the automaton that transitions to one of
$\{q_1,q_2,q_3\}$ upon reading the symbol $a$.
Then, starting from $q$ we can
say that a word $a.w$ would be accepted if it accepts $w$ from either $q_1$ or
$q_2$ or $q_3$.
Briefly, we can express that as: $\mathrm{acc}(q,aw) =
  \mathrm{acc}(q_1,w) \vee\mathrm{acc}(q_2,w) \vee \mathrm{acc}(q_3,w)$.
In an AFA, this acceptance condition can be expressed as an {\em arbitrary}
combination of conjunctions or disjunctions.
For example $\mathrm{acc}(q,aw) =
  (\mathrm{acc}(q_1,w)\wedge\mathrm{acc}(q_2,w))\vee\mathrm{acc}(q_3,w)$ means
that the automaton accepts from $q$ if it accepts $w$ from {\em both} $q_1$ and
$q_2$ or it accepts from $q_3$.
AFA are a well-studied artifact, see
\cite{vardi1995alternating,vardi2005alternating} for a detailed treatment.

In this paper, we show that discrete-time STREL formulas can be converted into
AFA.
This translation is done in two steps: (1) We first convert a discrete-time
STREL formula into an equivalent formula that does not use temporal operators
with intervals.
As we assume full knowledge of the discrete time-stamps, for
example, a formula $\mathbf{F}_{[a,b]} \varphi$ can be converted into a
disjunction over $\mathbf{X}$ (next time) formulas.
(2) We then convert the
resulting formula into an alternating finite automaton.
We prove that the
complexity of the resulting AFA is linear in the size of the intermediate
formula (with $\mathbf{X}$ operators).
% \textcolor{red}{add the other bit if needed}.

A recent development in automata-based monitoring and quantitative satisfaction
semantics for temporal logics is the use of a general algebraic framework to
define semantics of the automaton/logic.
The observation is that the usual
qualitative (Boolean) satisfaction semantics for most temporal logics utilize
the Boolean semiring $(\vee,\wedge, 0, 1)$.
However, by allowing more general
semirings, a logical formula can be re-interpreted to produce some value in the
chosen semiring \cite{jakvsic2018algebraic,nenzi2022logic}.
For example, the
semiring $(\max,\min,-\infty,+\infty)$ recovers the quantitative semantics of
Metric Temporal Logic as defined in \cite{fainekos2009robustness}.
An attractive
aspect of STREL is that it also provides quantitative semantics based on
semiring semantics.
In this paper, we show how the semiring semantics of STREL
can be directly incorporated into the AFA: instead of a standard AFA, we instead
get a {\em weighted} AFA.
The weighted AFA is especially useful for performing
quantitative monitoring in both online and offline settings.

% \todotext{Add bit about other spatial logics.}
% 
% This logic allows for specifications that can capture complex fractal spatial patterns. However, manually finding the correct formulation can be highly challenging.
% 
% Signal Spatio-Temporal Logic (SSTL)~\cite{Nenzi18_SSTL} extends STL with the somewhere and the surround operators and considers fixed spatial configurations.
% 
% Spatial Aggregation Signal Temporal Logic (SaSTL)~\cite{Ma21} enriches STL for
% expressing spatial aggregation and spatial counting characteristics. Similarly as SSTL, it considers only static graphs.
% 
% In this paper, we consider Spatio-Temporal Reach and Escape Logic (STREL)
% \cite{nenzi2022logic,bartocci2017monitoring}, a logic for monitoring distributed
% systems where various agents in the system are modeled as spatially dynamic
% graphs.

% \todotext{Why STREL?}
% In particular, STREL considers two new spatial operators, the reach and the escape, that allow us to specify properties over spatial paths. Furthermore, STREL can operate over a dynamic topological space.
% Another distinguishing feature of STREL is its ability to define semantics through the use of constraint semirings in algebraic structures.

% \todotext{Add a bit about AFA.}

To summarize, our main contributions in this paper are:
\begin{itemize}
  \item We present a linear-time construction for STREL specifications
        to alternating finite automaton (AFA).
  \item We additionally show how to adapt the quantitative semantics for STREL
        presented in \cite{nenzi2022logic} into \emph{weighted} AFA.
  \item We present how to use this construction for offline and online
        monitoring.
  \item We present results from applying the automaton-based monitoring approach for STREL on case studies in the SwarmLab environment.
        In these case studies, we monitor drone trajectories across various map configurations to evaluate their compliance with reach-avoid specifications, ensuring that drones maintain contact with the flock and can reconnect promptly when temporarily disconnected.
\end{itemize}

% !TEX root = ../main.tex

\section{Preliminaries}

\begin{definition}[Semiring~\cite{kuich1986semirings,golan1999semirings}]
  A tuple, \(\Ke = \Tuple{K, \oplus, \otimes, \bot, \top}\) is
  a \emph{semiring} with the underlying set \(K\) if \(\left(K, \oplus,
  \bot\right)\) is a commutative monoid with identity \(\bot\);
  \(\left( K, \otimes, \top \right)\) is a monoid with identity element
  \(\top\); \(\otimes\) distributes over \(\oplus\); and \(\bot\) is
  an annihilator for \(\otimes\) (for all \(k \in K, k \otimes \bot =
  \bot \otimes k = \bot\)).
\end{definition}

A semiring \(K\) is called a \emph{commutative} if multiplication \(\otimes\)
is also commutative and \emph{simple} if \(k
\oplus \top = \top\), for all \(k \in K\).
A semiring \(K\) is additively (or multiplicatively) \emph{idempotent} if, for
all \(k \in K\), \(k \oplus k = k\) (or \(k \otimes k = k\)).
A semiring that is commutative, additively idempotent, and simple is known as a
\emph{constraint semiring} or a \emph{c-semiring}.
Additionally, if a c-semiring is also multiplicatively idempotent, the semiring
creates a \emph{bounded distributive lattice}, where \(\top\) is the supremum
element and \(\bot\) is the infimum element.

\begin{definition}[De Morgan Algebras~\cite{cignoli1983dualities}]
  A De Morgan algebra is a structure \((K, \oplus, \otimes, \ominus, \bot, \top)\) such that:
  \begin{itemize}
    \item \((K, \oplus, \otimes, \bot, \top)\) is a \emph{bounded
            distributive lattice}; and
    \item \(\ominus\) is a negation function on \(K\) such that, for \(a,b \in K\)
          \begin{itemize}
            \item \(\ominus (a \oplus b) = \ominus a \otimes \ominus b\),
            \item \(\ominus (a \otimes b) = \ominus a \oplus \ominus b\),
            \item \(\ominus \ominus a = a\),
          \end{itemize}
          i.e., \(\ominus\) is an \emph{involution} on \(K\) that follows De
          Morgan's laws.
  \end{itemize}
\end{definition}

Examples of De Morgan algebras include
\begin{itemize}
  \item The Boolean algebra \((\Be, \lor, \land, \neg, \bot, \top)\) with
        \(\land, \lor,\) and \(\neg\) referring to the usual logical
        conjunction, disjunction, and negation respectively.
  \item The min-max algebra over the extended reals \((\Re^\infty = \Re \cup
        \{\infty, -\infty\}, \max{}, \min{}, -, \bot, \top)\).
\end{itemize}

\begin{definition}[Distance Domain]
  \((D, \bot_D, \top_D, +_D, \leq_D)\) is a \emph{distance domain} such that
  \(\leq_D\) is a total ordering over \(D\) with \(\bot_D\) and \(\top_D\) as
  infimum and supremum respectively, and \((B, \bot_B, +_B)\) is a monoid.
\end{definition}
The following are examples of distance domains:
\begin{itemize}
  \item Counting domain \((\Ne^\infty, 0, \infty, +, \leq)\), where \(\Ne^\infty
        = \Ne \cup \left\{ \infty \right\}\) is the extended natural numbers, \(+\)
        is the usual addition with identity \(0\) and bounded at \(\infty\) (\(i
        + \infty = \infty\)), and \(\leq\) is the usual definition of total
        ordering.
  \item Tropical domain \((\Re^+, 0, \infty, +, \leq)\) where \(\Re^+
        = \Set{d \in \Re \given d \geq 0 } \cup \infty\) is the extended positive
        reals, with \(+\) and \(\leq\) defined as usual.
\end{itemize}

\subsection{Spatial Models}

\begin{definition}[Spatial Model \cite{nenzi2022logic}]
  Given some set \(W\), a \(W\)-spatial model \(\Sc\) is a graph \((L, E)\)
  where \(L\) is a finite set of \emph{locations}, also called a \emph{space
    universe}, and \(E \subseteq L \times W \times L\) is a finite set of labeled
  edges between the locations.
\end{definition}
We will use \(\Se_W\) to describe the set of all \(W\)-spatial models, while
\(\Se_W^L\) indicates the set of \(D\)-spatial models having \(L\) as a set of
locations.
For a given \(\Sc \in \Se_W^L\), we abuse notation and let \(\Sc(l_1, l_2) = w\)
refer to the edge \((l_1, w, l_2)\) in \(\Sc\).

\begin{definition}[Paths in \(\Sc\)]
  Let \(\Sc= (L,E) \in \Se_W^L\) be a \(W\)-spatial model.
  Then, a \emph{path} \(\tau\) is a sequence of locations \(l_0 l_1 \cdots\)
  such that \(l_i \in L\) and \((l_i, w, l_{i+1}) \in E\) for some \(w \in W\)
  and all \(i \in \Ne\), and additionally, no edge \((l_i, w, l_{i+1})\) occurs
  twice.
  A \emph{simple path} is a path where no location repeats.
\end{definition}

Note that since an edge cannot occur more than once in a path, the number of
paths in a spatial model is finite, and the length of a path \(\tau\) (denoted
\(\abs{\tau}\)) is also finite.
We let \(Routes(\Sc)\) denote the set of routes in \(\Sc\) and \(Routes(\Sc,
l)\) denote the set of routes starting at location \(l\).

For a path \(\tau = l_0 l_1 \cdots\), \(i \in \Ne\), and \(l \in L\), we use
\begin{itemize}
  \item \(\tau[i]\) to denote the \(i\)th node \(l_i\) in \(\tau\)

  \item \(\tau[..i]\) to indicate the prefix \(l_0 l_1 \cdots l_i\), and
        \(\tau[i..]\) to indicate the suffix \(l_i l_{i+1} \cdots\), both of
        which are also paths in \(Routes(\Sc)\)

  \item \(l \in \tau\) when there exists an index \(i\) such that \(\tau[i] =
        l\), and \(l \not\in \tau\) otherwise.
\end{itemize}

\begin{definition}[Distance functions and Distance over paths]
  Let \(\Sc = (L, E)\) be a \(W\)-spatial model and \((D, \bot_D, \top_D, +_D,
  \leq_D)\) be a distance domain.
  A \emph{distance function} \(f: W \to D\) is a mapping from an edge
  weight in \(\Sc\) to the distance domain, and the distance over a path
  \(d^f: Routes(\Sc) \to D\) is defined as:

  \begin{displaymath}
    d^f(\tau) =
    \begin{cases}
      \sum_{i=1}^{\abs{\tau}} f(\Sc(\tau[i-1], \tau[i])) & ~\text{if}~ \abs{\tau} \geq 2 \\
      \bot_B                                             & ~\text{otherwise},
    \end{cases}
  \end{displaymath}
  where the summation is defined over the monoid \(+_D\).
\end{definition}

We use the notation \(d^f_\tau[i]\) to refer to \(d^f(\tau[..i])\), i.e., the
path distance of the prefix of a path \(\tau\) up to the \(i\)th location on
the path.
And we use the following notation to describe the distance shortest path
between two locations \(l_1, l_2 \in L\) in \(\Sc\):
\begin{displaymath}
  d^f_{\Sc}[l_1,l_2] = \min\Set{d^f_\tau[i] \given \tau \in Routes(\Sc,l_1),
    \tau[i] = l_2}.
\end{displaymath}

\subsection{Spatio-temporal Reach and Escape Logic (STREL)}

The syntax for STREL~\cite{bartocci2017monitoring,nenzi2022logic} for
discrete-time traces is given by the recursive definition
\begin{equation}
  \begin{split}
    \varphi := &
    \mu \mid \neg \varphi \mid \varphi \land \varphi \mid \varphi \lor \varphi \\
               &
    \mid \TNext \varphi \mid \varphi \Until_{[t_1, t_2]}
    \varphi \mid \varphi \Reach_{[d_1, d_2]}^{f} \varphi
    \mid \Escape_{[d_1, d_2]}^{f} \varphi.
  \end{split}
\end{equation}
Here, \(\mu\) is an atomic predicate, negation \(\neg\), conjunction \(\land\),
and disjunction \(\lor\) are the standard Boolean connectives, \(\TNext\) is the
temporal \emph{Next} operator, and \(\Until_{[t_1, t_2]}\) is the \emph{Until}
temporal operator with a positive real time interval \([t_1, t_2]\).
These are the same temporal operators found in Linear Temporal Logic and Metric
Temporal Logic.
\emph{Reach} \(\Reach\) and \emph{Escape} \(\Escape\) are the \emph{spatial}
operators defined in STREL, where \(d_1, d_2 \in D\) are \emph{distances} in
a domain \(D\) and \(f\) is a distance function mapping to \(D\).
% \myhighlight{%
% In this paper, we restrict ourselves to spatial operators with the
% distance intervals starting at \(0\), i.e., \(d_1 = 0\) in the syntax above.
% Thus, we will use the syntax \(\phi_1 \Reach^f_{\leq d} \phi_2\) to simplify
% the notation.
% %
% }

Informally, the \emph{Reach} operator \(\phi_1 \Reach_{[d_1,d_2]}^f \phi_2\)
describes the behavior of reaching a location satisfying \(\phi_2\) through a
path \(\tau\) such that \(d_1 \leq d^f(\tau) \leq d_2\) and \(\phi_1\) is
satisfied at all locations in the path.
The \emph{Escape} operator \(\Escape_{[d_1, d_2]}^f \phi_1\) instead describes
the possibility of escaping from a certain region via a path that passes only
through locations that satisfy \(\phi_1\) such that the distance between the
starting location and the last location belongs in the interval \([d_1, d_2]\).
Note that for \emph{Escape}, the distance constraint is not necessarily over the
satisfying path, but the shortest distance path between the start and end
location (but they may be the same).

From the above syntax we can derive the standard temporal
operators \emph{Eventually} \(\Ev_{[t_1,t_2]}\) and \emph{Globally}
\(\Alw_{[t_1, t_2]}\).
Moreover, we can derive from them additional \emph{spatial operators}:
\begin{itemize}
  \item \(\diamonddiamond_{\leq d}^f \varphi = \top \Reach_{\leq d}^f \varphi\)
        denotes the \emph{Somewhere} operator
  \item \(\boxbox_{\leq d}^f \varphi = \neg \diamonddiamond_{\leq d}^f \neg
        \varphi\) denotes the \emph{Everywhere} operator
  \item \(\varphi_1 \varocircle_{\leq d}^f \varphi_2 \) denotes the \emph{Surround}
        operator.
\end{itemize}
The \emph{Somewhere} and \emph{Everywhere} operations describe behavior for some
or all locations within reach of a specific location, and the \emph{Surround}
expresses the notion of being surrounded by a region that satisfies \(\phi_2\),
while being in a \(\phi_1\) satisfying region.

To formally define the semantics of a STREL specification, we first define
a \emph{trace} over \(W\)-spatial models.

\begin{definition}[\(K\)-Labeled Trace]
  Let \((K, \oplus, \otimes, \ominus, \bot, \top)\)  be a De Morgan algebra and
  \(L\) a space universe.
  Then, a \emph{\(K\)-labeled trace} is a pair
  \((\sigma, \nu)\) such that
  \begin{itemize}
    \item \(\sigma = \Sc_0 \Sc_1 \cdots\) is a finite
          sequence of spatial models

    \item \(\nu(\Sc_i, l, \mu) \in K\) is a labeling
          function that maps the location \(l \in L\) in \(\Sc_i\) to a value in
          \(K\) for an atomic predicate \(\mu\).
  \end{itemize}
\end{definition}
We will simply refer to a trace as \(\sigma\) when the labeling function \(\nu\)
is obvious.

Thus, the algebraic semantics \cite{nenzi2022logic} \(\rho_l(\sigma, \phi, t)\)
of a discrete-time STREL specification \(\phi\) at an \emph{ego} location \(l
\in L\) for a \(K\)-labeled trace \(\sigma = \Sc_0 \Sc_1 \cdots\) at time \(t
\in \Ne\) can be defined as:

\begin{equation*}
  \begin{array}{lcl}
    \rho_l(\sigma, \mu, t)                               & = & \nu(\Sc_t, l, \mu)                                          \\
    \rho_l(\sigma, \neg \phi, t)                         & = & \ominus \rho_l(\sigma, \phi, t)                             \\
    \rho_l(\sigma, \phi_1 \land \phi_2, t)               & = & \rho_l(\sigma, \phi_1, t) \otimes \rho_l(\sigma, \phi_2, t) \\
    \rho_l(\sigma, \phi_1 \lor \phi_2, t)                & = & \rho_l(\sigma, \phi_1, t) \oplus \rho_l(\sigma, \phi_2, t)  \\
    \rho_l(\sigma, \phi_1 \Until_{[t_1,t_2]} \phi_2, t)  & = &
    \bigoplus\limits_{\mathclap{t' \in t + [t_1, t_2]}}
    \rho_l(\sigma, \phi_2, t') \otimes
    \bigotimes\limits_{\mathclap{t''\in [t',t]}} \rho_l(\sigma, \phi_1, t'')                                               \\
    \rho_l(\sigma, \phi_1\Reach_{[d_1,d_2]}^f \phi_2, t) & = &                                                             \\
    \multicolumn{3}{r}{%
      \bigoplus\limits_{{\tau \in Routes(\Sc_t, l)}}
      \bigoplus\limits_{{i: d^f_\tau[i] \in [d_1,d_2]}}
      \left(
      \rho_{\tau[i]}(\sigma, \phi_2, t)
      \otimes \bigotimes\limits_{j<i} \rho_{\tau[j]}(\sigma, \phi_1, t)
    \right) }                                                                                                              \\[2em]
    \rho_l(\sigma, \Escape_{[d_1,d_2]}^f \phi, t)        & = &                                                             \\
    \multicolumn{3}{r}{%
    \bigoplus\limits_{\tau \in Routes(\Sc_t,l)}
    \bigoplus\limits_{\substack{l' \in \tau                                                                                \\
    d^f_{\Sc_t}[l,l'] \in [d_1, d_2]}}
    \bigotimes\limits_{\substack{i \leq j                                                                                  \\
    \tau[j] = l'}}
    \rho_{\tau[i]}(\sigma, \phi, t)
    }
  \end{array}
\end{equation*}

See that for the above algebraic semantics, using the Boolean algebra \((\Be,
\lor, \land, \neg, \bot, \top)\) gets us the Boolean/quantitative semantics of
STREL over finite traces.
For the Boolean algebra, we say that a trace \(\sigma\) satisfies a
specification \(\phi\) at the ego location \(l\) (denoted \(\sigma \models
(\phi, l)\)) if the output of \(\rho_l(\sigma, \phi, 0)\) is \(\top\).
We also use the same notation, \(\Sc \models (\phi, l)\), to denote if the
specification \(\phi\) is satisfied for a spatial model \(\Sc\) at location
\(l\).

\subsection{Alternating Automata}

For a given set \(X\), let \(\PBool[X]\) be the set of positive Boolean formulas
over \(X \cup \{\top, \bot\}\), i.e., Boolean formulas built over \(X\) using
only \(\land\) and \(\lor\) logical connectives.
We let \(X^*\) denote the set of all finite-length sequences \(x_0,x_1,\ldots\) such that \(x_i \in X\) for all \(i \in \Ne\).

In a non-deterministic automaton \(\Ac = \left( \Sigma, Q, Q_0, \Delta, F
\right)\), with input alphabet \(\Sigma\), states \(Q\), initial states
\(Q_0 \subseteq Q\), and final states \(F \subseteq Q\), we usually model the
transition relation \(\Delta\) as a partial function \(\Delta: Q
\times \Sigma \to 2^Q\).
For example, for a state \(q \in Q\) and an input symbol \(a \in \Sigma\), the
transition relation may output the set of successor states as \(\Delta(q,
a) = \left\{ q_1, q_2, q_3 \right\}\).
The above can also be represented as a function \(\Delta: Q\times\Sigma \to
\PBool[Q]\) such that the set \(\left\{ q_1,q_2,q_3 \right\}\) is equivalent to the
formula \(q_1 \lor q_2 \lor q_3\).

While non-deterministic automata restrict such formulas to disjunctions only, in
alternating automata, we allow the use of arbitrary formula from \(\PBool[Q]\).
Thus, we can have an \emph{alternating} transition \[\Delta(q,a) = (q_1 \land
  q_2) \lor (q_3 \land q_4),\] which means that the automaton accepts a word
\(aw\), where \(a \in \Sigma, w \in \Sigma^*\) from state \(q\) if it accepts
the word \(w\) from both \(q_1\) and \(q_2\) or from both \(q_3\) and \(q_4\).
Such a transition combines \emph{existential} transitions (over the
disjunctions) and \emph{universal} transitions (over the conjunctions).

\begin{definition}[Alternating Automata~\cite{vardi1996automatatheoretic}]%
  \label{def:afa}
  An \emph{alternating automaton} is a tuple \(\Ac = \left( \Sigma, Q, q_0,
  \Delta, F \right)\), where \(\Sigma\) is a finite nonempty alphabet, \(Q\) is
  a finite set of states with initial state \(q_0 \in Q\), \(F \subseteq Q\)
  is a set of accepting states, and \(\Delta: Q \times \Sigma \to \PBool[Q]\) is
  a transition relation function.
\end{definition}

Due to the nature of the universal transitions, for an input \emph{word} \(\sigma \in \Sigma^*\), where \(\sigma = \sigma_0 \sigma_1 \ldots\), the \emph{run} \(\Run_\Ac(\sigma)\) of an alternating automaton \(\Ac\) induced by \(\sigma\) is a \(Q\) labeled rooted tree where
\begin{itemize}
  \item The root of the tree is labeled by \(q_0\)
  \item If an interior node, at some level \(i\), has only disjunctions in \(\theta' = \Delta(q_i, \sigma_i)\), the node has a single child labeled by a \(q\) in \(\theta\).
  \item If an interior node, at some level \(i\), has a conjunction over \(k\) disjunction clauses, then it will have \(k\) children, each labeled by one state in the \(k\) clauses.
\end{itemize}
A run is \emph{accepting} if all the leaf nodes are labeled by states in \(F\).
Intuitively, the automaton monitor clones itself and transitions to all the states in a universal transition, and the runs of each cloned automata must all be accepting for the word to be accepted.

% !TEX root = ../main.tex

\section{Conversion to Alternating Automata}

% \begin{remark}
%   The following construction holds specifically for discrete-time semantics.
%   Thus, the \emph{next} operator \(\TNext \phi\) is well-defined.
% \end{remark}

In this section, we will derive a framework for translating STREL specifications to alternating word automata.
We will first look at translating STREL to an intermediate logic, \emph{Spatial Linear Temporal Logic} (SpLTL), from which we will derive the direct translation into an alternating word automaton.

\subsection{Spatial Linear Temporal Logic}

Similar to STREL, we define the syntax of a Spatial Linear Temporal Logic (SpLTL) recursively as:
\begin{equation}
  \begin{split}
    \varphi := &
    \mu \mid \neg \varphi \mid \varphi \land \varphi \mid \varphi \lor \varphi \\
               &
    \mid \TNext \varphi \mid \varphi \Until
    \varphi \mid \varphi \Reach_{[d_1, d_2]}^{f} \varphi
    \mid \Escape_{[d_1, d_2]}^{f} \varphi.
  \end{split}
\end{equation}

Here, the various operators maintain the same semantics as in discrete-time
STREL, except that we do not permit intervals on \emph{Until} (and by extension
on the derived \emph{Eventually} and \emph{Globally} operators).
For the SpLTL operators, we retain the same notation for the semantics \(\rho\)
and for denoting accepting traces and models.

\begin{lemma}
  \label{lemma:strel-spltl}
  For any discrete-time STREL specification, \(\varphi\), there exists a corresponding SpLTL specification \(\varphi'\) such that
  \begin{enumerate}
    \item For an input trace \(\sigma\) and ego location \(l\), \(\sigma \models (\varphi', l)\) \emph{if and only if} \(\sigma \models (\varphi, l)\).
    \item \(\abs{\varphi'}\) is in \(O(T)\), where \(\abs{\cdot}\) counts the number of syntactic elements and terms in the formula and \(T\) is the sum of all the intervals in \(\varphi\).
  \end{enumerate}
\end{lemma}
\begin{proof}
  We can prove the above by showing the expansion formulas for every temporal operator with an interval to an untimed one~\cite{donze2013efficient,esparza2016ltl,vardi1996automatatheoretic}:
  \begin{itemize}
    \item \(\phi \Until_{[t_1, t_2]} \psi \equiv \Ev_{[t_1, t_2]} \psi \land
          \left(\phi \Until_{[t_1, \infty)} \psi \right)\), for \(t_1, t_2 \in
          \Ne, t_1 < t_2\)
    \item \(\phi \Until_{[t_1, \infty)} \psi \equiv \Alw_{[0,t_1]} \left( \phi
          \Until \psi\right)\)
          % \item \(\phi \Until \psi \equiv \psi \lor \left(\phi \land \TNext \left( \phi
          %       \Until \psi\right)\right)\)
    \item \(\Ev \phi \equiv \top \Until \phi\) and \(\Alw \phi \equiv \neg \Ev
          \neg \phi \equiv \neg \left( \top \Until \neg\phi \right)\)
    \item Bounded formulas \(\Ev_{[t_1, t_2]} \phi \) and \(\Alw_{[t_1,t_2]}
          \phi\), where \(t_1,t_2 \in \Ne, t_1<t_2\), are equivalent to nested
          disjunctions or conjunctions of the \emph{next} operator.
          For example,
          \[F_{[0,3]} a \equiv a \lor \TNext (a \lor \TNext (a \lor \TNext a)).
          \]
  \end{itemize}

  From the above, we can see that for each temporal expression, there is a linear increase in the number of syntactic elements.
  Specifically, if an interval for a formula is \([t_1, t_2]\), we see that the formula increases by a factor of \(t_2 - t_1\) (with some constant overhead).
  Thus, if \(T\) is the sum of all the interval sizes in \(\varphi\), the size of \(\varphi'\) is in \(O(T)\).
\end{proof}

% Let us first establish some expansion formulas for the temporal
% operators~\cite{donze2013efficient,esparza2016ltl,vardi1996automatatheoretic}:
% \begin{itemize}
%   \item \(\phi \Until_{[t_1, t_2]} \psi \equiv \Ev_{[t_1, t_2]} \psi \land
%         \left(\phi \Until_{[t_1, \infty)} \psi \right)\), for \(t_1, t_2 \in
%         \Ne, t_1 < t_2\)
%   \item \(\phi \Until_{[t_1, \infty)} \psi \equiv \Alw_{[0,t_1]} \left( \phi
%         \Until \psi\right)\)
%   \item \(\phi \Until \psi \equiv \psi \lor \left(\phi \land \TNext \left( \phi
%         \Until \psi\right)\right)\)
%   \item \(\Ev \phi \equiv \top \Until \phi\) and \(\Alw \phi \equiv \neg \Ev
%         \neg \phi \equiv \neg \left( \top \Until \neg\phi \right)\)
%   \item Bounded formulas \(\Ev_{[t_1, t_2]} \phi \) and \(\Alw_{[t_1,t_2]}
%         \phi\), where \(t_1,t_2 \in \Ne, t_1<t_2\), are equivalent to nested
%         disjunctions or conjunctions of the \emph{next} operator.
% \end{itemize}

% Given a STREL specification \(\phi\), let \(\abs{\phi}\) be the size of the
% formula after performing the above expansions such that the only operators in
% the expanded formula are the propositional operations \(\neg\), \(\land\), and
% \(\lor\), and the temporal operators \(\TNext\) and \(\Until\).

\subsection{SpLTL to AFA}

We use \(\Phi_\phi\) to denote the set of all subformulas in an SpLTL specification \(\phi\) and their negations (including \(\phi\) and \(\neg\phi\)).

Let \(A\) be the finite set of atomic predicates in the system and \(L\) be the
space universe.
A labeling function \(\nu: \Se_W^L \times L \times A \to \Be\) is a mapping such
that \(\nu(\Sc, l, \mu)\) determines if the predicate \(\mu\) is satisfied at
location \(l\), i.e., if \(\Sc \models (\mu, l)\).

From these, we construct an automaton \(\Ac = (\Sigma, Q, q_0, \Delta, F)\)
such that \(\Sigma \subseteq \Se_W^L\) is the input alphabet over dynamic
spatial models with universe \(L\):

\paragraph*{Automaton states \(Q\)}
We define the set of states \(Q\) in the automaton to represent each of the
formulas in \(\Phi_\phi\) evaluated at each location in \(L\) in the system
\(Q = \Set{\forall \psi \in \Phi_\phi, \forall l \in L, q_\psi^l }.
\)
The initial state corresponding to an ego location \(l_0\) is the state
corresponding to the formula we are monitoring \(\phi\), i.e., \(q_\phi^{l_0}\).
The accepting states \(F\) consist any
state representing a formula of the form \(\neg(\phi \Until \psi)\) at any
location (by extension, this would include states corresponding to unbounded
\emph{Globally} expressions \(\Alw \varphi\)).
Moreover, the Boolean formula \(\top \in \PBool[Q]\) is an accepting sink state.

\begin{note}
  We will eschew the superscript location over a state \(q\) when the context is
  clear that the superscript does not change.
  Moreover, the initial state of the automaton can be picked depending on the
  ego location such that the initial state is always \(q_\phi^{l_0}\), where
  \(\phi\) is the specification being monitored and \(l_0\) is the choice of ego
  location.
\end{note}

For any formula \(\theta \in \PBool[Q]\), the \emph{dual} \(\overline{\theta}\) of the
formula can be obtained as follows:
% LTeX: enabled=false
\begin{equation*}
  \begin{array}{c@{~=~}l}
    \overline{q}_\top^l      & \bot                                                     \\
    \overline{q}_\bot^l      & \top                                                     \\
    \overline{q}_{\varphi}^l & q_{\neg\varphi}^l \text{~for~} q_\varphi^l \in Q         \\
    \overline{\left(q_\varphi^{l_1} \land q_\psi^{l_2}\right)}
                             & \overline{q}_\varphi^{l_1} \lor \overline{q}_\psi^{l_2}
    \\
    \overline{\left(q_\varphi^{l_1} \lor q_\psi^{l_2}\right)}
                             & \overline{q}_\varphi^{l_1} \land \overline{q}_\psi^{l_2}
  \end{array}
\end{equation*}
% LTeX: enabled=true

\paragraph*{Transition relation \(\Delta\)}

For some \(\Sc \in \Sigma\), we define the transition relation \(\Delta: Q \times
\Sigma \to \PBool[Q]\) recursively on states corresponding to the subformulas in
the specification \(\phi\) as follows:
\begin{equation*}
  \renewcommand{\arraystretch}{2.2}
  \begin{array}{l@{\ =\ \ }l}
    \Delta(q_\top^l, \cdot)             & \top                                                \\
    \Delta(q_\bot^l, \cdot)             & \bot                                                \\
    \Delta(q_\mu^l, \Sc)                & \nu(\Sc, l, \mu)                                    \\
    \Delta(q_{\neg \phi}^l, \Sc)        & \overline{\Delta(q_\phi^l, \Sc)}                    \\
    \Delta(q_{\phi \land \psi}^l, \Sc)  & \Delta(q_\phi^{l}, \Sc) \land \Delta(q_\psi^l, \Sc) \\
    \Delta(q_{\TNext\phi}^l, \Sc)       & q_\phi^l                                            \\
    \Delta(q_{\phi \Until \psi}^l, \Sc) & \Delta(q_\psi^l,\Sc) \lor \left(
    \Delta(q_\phi^l, \Sc) \land q_{\phi \Until \psi}^l\right)                                 \\
    % \Delta(q_{\Ev \phi}^l, \Sc)             & \Delta(q_\phi^l) \lor q_{\Ev\phi}^l             \\
  \end{array}
\end{equation*}
\begin{equation*}
  \renewcommand{\arraystretch}{2.2}
  \begin{array}{c@{\ =\ }l}
    \Delta(q_{\phi \Reach^f_{[d1,d2]} \psi}^l, \Sc) & \\
    \multicolumn{2}{c}{%
      \bigvee\limits_{\tau \in Routes(\Sc, l)}
      \bigvee\limits_{i: \left(d_\tau^f[i]\in [d_1,d_2]\right)}
      \left( \Delta(q_\psi^{\tau[i]}, \Sc) \land \bigwedge\limits_{j < i} \Delta(q_{\phi}^{\tau[j]},\Sc) \right)
    }                                                 \\
    \Delta(q_{\Escape^f_{\leq d} \psi}^l, \Sc)      & \\
    \multicolumn{2}{c}{%
    \bigvee\limits_{\tau \in Routes(\Sc_t,l)}
    \bigvee\limits_{\substack{l' \in \tau             \\
    d^f_{\Sc_t}[l,l'] \in [d_1, d_2]}}
    \bigwedge\limits_{\substack{i \leq j              \\
    \tau[j] = l'}}
    \Delta(q_\psi^{\tau[i]}, \Sc)
    }
  \end{array}
\end{equation*}

% \begin{todocomment}
%   \begin{itemize}
% \item Reword the reach operator similar to the Until operator (essentially a recursive formula).
% \item Relax the intervals on the Reach operator to be \(\leq d\) similar to Somewhere and Everywhere.
% \end{itemize}
% \end{todocomment}

\begin{theorem}[Linear-sized SpLTL automata]
  Given a SpLTL formula \(\phi\) with atomic propositions \(A\) and locations
  \(L\), one can build an alternating finite automaton \(\Ac = \Tuple{\Sigma,
    Q, q_0, \Delta, F}\) where the language \(\Lc(\Ac, l)\) is exactly the set of
  traces that satisfy the formula \(\phi\) at \emph{ego location} \(l\) such
  that \(\abs{Q}\) is in \(O(\abs{L}\abs{\phi})\).
\end{theorem}

\begin{proof}
  We can clearly see from the above construction that \(Q\) has states
  corresponding to each subformula in \(\phi\), their duals, and states
  corresponding to \(\top\) and \(\bot\).
  Moreover, since \(\Delta\) splits the computation across locations, we
  have, for each location \(l \in L\), a node corresponding to each subformula
  and its dual.
  Thus, \(Q\) is such that \(\abs{Q} \leq 2 \abs{L} \abs{\phi}\).
\end{proof}

\begin{remark}
  Note that while, in general, the \(Q\) can have \(2 \abs{L}
  \abs{\phi}\) states, for several practical STREL formulas, many states
  corresponding to subformulas in \(\Phi_\phi\) are not reachable from any
  initial state \(q_\phi^l\).
  These states can be determined purely through the set of reachable states
  through the temporal operators, and the non-reachable states can be pruned.
\end{remark}

% !TEX root = ../main.tex

\subsection{Alternating Weighted Automata}

In this section, we look at how we can generalize the above construction to
incorporate the quantitative semantics of STREL.

Notice that for some finite set \(X\), formulas \(\theta \in \PBool[X]\) that
are in DNF form are effectively polynomials over the Boolean semiring with
support variables \(X\), with logical \emph{and} as multiplication and logical
\emph{or} as addition.

For a commutative semiring \(K\) and a set \(X = \left\{ x_1, \ldots, x_n
\right\}\), we use \(K[X]\) to denote the set of all polynomials with
indeterminates \(x_1,\ldots,x_n\) and coefficients in \(K\).
A monomial \(m\) with indeterminates \(X = \left\{x_1, \ldots, x_n\right\}\) is
a polynomial in \(K[X]\) such that \(m = c x_1^{k_1} x_2^{k_2} \ldots
x_n^{k_n}\) for some coefficient \(c \in K\) and non-negative integers
\(k_1,\ldots,k_n\) where multiplication is under the semiring operation.
A constant \(a \in K\) is a monomial in \(K[X]\) with the power of all
indeterminates equal to 0 and coefficient \(a\).
The \emph{support} of a polynomial is the set of indeterminates (or variables)
with coefficients not equal to \(\bot \in K\).
Moreover, a polynomial \(\theta \in K[X]\) is a scalar function \(\theta:
K^{\abs{X}} \to K\).

For idempotent semirings (or, equivalently, bounded distributive lattices),
since \(x^k = x\) for all non-negative integers \(k\) such that \(k > 0\), we
can see that the \emph{degree} of the polynomials \(K[X]\) over a bounded
distributive lattice does not exceed \(1\).
This includes the Boolean semiring (equivalent to the previous section) and the
\(\min-\max\) semiring, which defines the quantitative semantics of STREL.

Let us now formally define the alternating weighted automaton, along with the
conversion of a STREL formula to such an automaton while preserving quantitative
semantics.

\begin{definition}[Alternating Weighted Automata~\cite{kostolanyi2018alternating}]
  Let \(K\) be a commutative semiring.
  Then, an \emph{alternating weighted automaton} over \(K\) is a tuple \(\Ac
  = \left( \Sigma, Q, \Delta, \alpha, \beta \right)\), where \(\Sigma\) and
  \(Q\) are the input alphabet and the non-empty finite set of states.
  Here, the transition function \(\Delta\) is a polynomial assigning function
  \(\Delta: Q \times \Sigma \to K[Q]\) with
  \(\alpha \in K[Q]\) as an initial weight polynomial and
  \(\beta: Q \to K\) as the terminal weighting function.
\end{definition}

While the definition of \(\Sigma\) remains the same as in the previous section,
for \(Q\), we will not populate it with any negations unless they also appear in
the specification \(\phi\).
The transition function \(\Delta\) needs to be redefined in terms of
polynomials over semirings, along with the initial polynomial and terminal
weights, \(\alpha\) and \(\beta\).
Moreover, the labeling function, \(\nu\) needs to be defined over the semiring
\(K\).

\paragraph*{Initial and Terminal Weights}

For a STREL specification \(\varphi\) and an \emph{ego} location \(l_0\), the
initial weight polynomial \(\alpha \in K[Q]\) is exactly the polynomial
\(q_\varphi^{l_0}\), i.e, \(\alpha\) is the polynomial with the initial state
corresponding to the ego location as the support and unitary coefficient
\(\top\).
The \emph{terminal weighting function}, \(\beta: Q \to K\), is a mapping such
that
\begin{displaymath}
  \beta(q) =
  \begin{cases}
    \top, & \text{~if~} q \in F, \\
    \bot, & \text{~otherwise.}
  \end{cases}
\end{displaymath}
These correspond to the initial states \(Q_0\) and final states \(F\)
respectively.

\begin{note}
  Similar to the Boolean case, the initial polynomial \(\alpha\) of the weighted
  automaton can vary depending on the choice of the ego location.
  Specifically, the initial polynomial is always the variable for state
  \(q_\phi^{l_0}\), where \(\phi\) is the specification being monitored and
  \(l_0\) is the ego location.
\end{note}

\paragraph*{Labeling function \(\nu\)}

Depending on the choice of algebra for \(K\), let \(\nu: \Se_W^L \times L \times
A \to K\) be a mapping of an atomic predicate from the spatial domain to \(K\).
An example can be a notion of a ``distance'' that a location \(l\) needs to move
for it to satisfy the predicate \(\mu\).
For example, as in \cite{fainekos2009robustness} and \cite{donze2010robust},
this can be a signed distance such that if the value is greater than 0, the
predicate is satisfied, and negative if the predicate is not, and can be used
with the min-max algebra.

\paragraph*{Transition relation \(\Delta\)}

We redefine the Boolean transition relation over an arbitrary semiring
polynomial \(K[Q]\).
Specifically, for some \(\Sc \in \Sigma\), we define the transition relation
\(\Delta: Q \times \Sigma \to K[Q]\) recursively on states
corresponding to subformulas in \(\Phi_\varphi\) as follows:
\begin{equation}
  \renewcommand{\arraystretch}{2.2}
  \begin{array}{c@{\ =\ }l}
    \Delta(q_\top^l, \cdot)                         & \top                                                  \\
    \Delta(q_\bot^l, \cdot)                         & \bot                                                  \\
    \Delta(q_\mu^l, \Sc)                            & \nu(\Sc, l, \mu)
    \\
    \Delta(q_{\neg \phi}^l, \Sc)                    & \ominus \Delta(q_\phi^l, \Sc)                         \\
    \Delta(q_{\phi \land \psi}^l, \Sc)              & \Delta(q_\phi^{l}, \Sc) \otimes \Delta(q_\psi^l, \Sc) \\
    \Delta(q_{\TNext\phi}^l, \Sc)                   & q_\phi^l                                              \\
    \Delta(q_{\phi \Until \psi}^l, \Sc)             & \Delta(q_\psi^l,\Sc) \oplus \left(
    \Delta(q_\phi^l, \Sc) \otimes q_{\phi \Until \psi}^l\right)                                             \\
    \Delta(q_{\phi \Reach^f_{[d1,d2]} \psi}^l, \Sc) &                                                       \\
    \multicolumn{2}{c}{%
      \bigoplus\limits_{\tau \in Routes(\Sc, l)}
      \bigoplus\limits_{i: \left(d_\tau^f[i]\in [d_1,d_2]\right)}
      \left( \Delta(q_\psi^{\tau[i]}, \Sc) \otimes \bigotimes\limits_{j < i} \Delta(q_{\phi}^{\tau[j]},\Sc) \right)
    }                                                                                                       \\
    \Delta(q_{\Escape^f_{\leq d} \psi}^l, \Sc)      &                                                       \\
    \multicolumn{2}{c}{%
    \bigoplus\limits_{\tau \in Routes(\Sc_t,l)}
    \bigoplus\limits_{\substack{l' \in \tau                                                                 \\
    d^f_{\Sc_t}[l,l'] \in [d_1, d_2]}}
    \bigotimes\limits_{\substack{i \leq j                                                                   \\
    \tau[j] = l'}}
    \Delta(q_\psi^{\tau[i]}, \Sc)
    }
  \end{array}
\end{equation}

\paragraph*{Run of an alternating weighted automaton}

Similar to a standard alternating automaton, the run of an alternating weighted
automaton is also a tree.
But, unlike the Boolean version, where the branches of the tree terminate in
Boolean values, determining the satisfaction value of a run, the branches of
a weighted run end in leaves taking value in a given algebra \(K\).
The \emph{weight} \(w_\Ac(\sigma)\) of a run induced by a trace \(\sigma\) on the automaton \(\Ac\) is the sum over paths in the run tree with the semiring \(K\), followed by a product over the value of each branch.

% \begin{todocomment}
% We should give a mathy definition of the run if time permits
% \end{todocomment}

% \begin{todocomment}
%   Need to write.
% \end{todocomment}

% !TEX root = ../main.tex

\section{Automata-based Monitoring for STREL}

Let the states in \(Q\) be consistently indexed such that \(q_1, \ldots,
q_{\abs{Q}} \in Q\).

As an alternating (weighted) automaton is characterized by polynomials over De
Morgan polynomials \(K[Q]\), we can see that to keep track of the current state
of the run of the automaton, it is sufficient to keep track of a single
polynomial.
Specifically, let \(\theta \in K[X]\) be the current state of the automaton.

At the start of the execution, for the specification \(\phi\) and an ego
location \(l_0\), \(\theta\) can be initialized to the state \(\alpha\)
corresponding to the given ego location \(l_0\).
Then for subsequent input symbol \(\Sc \in \Sigma\), we can compute the next
state of the automaton \(\theta'\), where
\begin{displaymath}
  \theta' \gets \theta\of*{\Delta(q_1, \Sc), \ldots, \Delta(q_{\abs{Q}}, \Sc)}.
\end{displaymath}
That is, for each state variable \(q\) in the support of \(\theta\), we evaluate
the transition function and substitute it into \(\theta\), producing a new
polynomial \(\theta' \in K[Q]\).

Then, to evaluate the satisfaction value or weight of a run at a given state
\(\theta\) at time \(t\), we simply use the final mapping \(\beta\) such that
\begin{displaymath}
  \rho_{l_0}(\sigma_0\cdots\sigma_t, \phi, 0) = \theta(\beta(q_1), \ldots, \beta(q_{\abs{Q}})).
\end{displaymath}
This method can be used both for online monitoring by keeping track of the
current state \(\theta\) and evaluating with \(\beta\) as needed, or in an
offline setting, by evaluating with \(\beta\) at the end of the trace.

\paragraph{Complexity of Monitoring Spatial Operators}

Notice that in the transition definition for nodes corresponding to \emph{Reach}
and \emph{Escape} operations, the \(Routes(\cdot, l)\) searches over all
possible routes originating from location \(l\).
In general, this can lead to an explosion in the search space.
Indeed, simply the problem of enumerating all paths between two locations
\(s,t\) in a graph is known to be \(\#P\)-complete \cite{valiant1979complexity}.

In \cite{nenzi2022logic}, the authors overcome this for offline monitoring as
they compute the satisfaction value for each subformula before evaluating the
spatial operators.
In our case, since we are consuming input \(\Sc \in \Se_W^L\) for every time
step, in general, we cannot realize whether a specific subformula for \(\Delta\)
holds true at a location before the end of the trace.

To this end, we adopt a depth-limited search algorithm to
enumerate paths from an ego location \(l\) \cite{agarwal2000depth}.
Here, the goal is to repurpose the standard depth-first search algorithm to enumerate each path such that we avoid repeating locations in a path by marking them visited, and then removing the mark just before popping from the stack.
% Moreover, we keep track of the length of the current path and only explore further if the distance threshold is met.

\begin{proposition}
  \label{lemma:depth-lim-reach}
  Given a spatial model \(\Sc = (L, E) \in \Se_W^L\), a distance function \(f:
  W \to D\), and an interval \([d_1, d_2]\), we define \(d_{min} = \min
  \Set*{f(w) \given (l,w,l) \in E }\) and \(k = \min\Set*{i \given i * d_{min} >
    d_2}\)
  (where \(i * d_{min}\) indicates the sum of \(i\) copies of \(d_{min}\).
  Let \(b = \max_l {(l, w, l') \in E}\) be the maximum node degree in \(\Sc\).
  Then, the complexity to enumerate the list of simple paths from a location
  \(l\) is in \(O(b^{d_{min}})\).
\end{proposition}

\paragraph{Complexity of updating states}

When working with Boolean polynomials, a common data structure to use is Binary
Decision Diagrams~\cite{bryant1992symbolic}, or Algebraic Decision Diagrams in
the context of finite domains \cite{bahar1997algebric}.
It is a well-established result that reduced-ordered binary decision diagrams
(ROBDD) are canonical representations of arbitrary Boolean expressions and that
the computing a \emph{restriction} (evaluating the expression with Boolean
values for each variable) or a \emph{composition} (substituting a variable with
another expression) can be done in time linear in the size of each expression.

Similar work has been done in the context of representing arbitrary semiring
polynomials as \emph{decision diagrams}
\cite{wilson2005decision,fargier2013semiring}, where polynomial operations like
restriction and composition are linear in the size of the polynomials.

\begin{lemma}
  \label{lemma:delta-poly}
  For any update to the current state \(\theta \in K[Q]\) of an alternating
  automaton \(\Ac = (\Sigma, Q, \Delta, \alpha, \beta)\) with state variables
  \(Q\), where \(K\) is a De Morgan algebra, the runtime for evaluating the next
  state is in \(O(\abs{Q}\abs{\theta})\).
  Moreover, evaluating the final weight \(\theta(\beta)\) of a run from a given state \(\theta \in K[Q]\) can be done in \(O(\abs{\theta})\) time.
\end{lemma}

\begin{lemma}
  \label{lemma:demorgan-poly}
  For a polynomial \(\theta \in K[X]\), where \(K\) is a De Morgan algebra and
  \(X\) is a finite set of variables, the number of terms in \(\theta\) (denoted
  \(\abs{\theta}\)) is at most \(2^{\abs{X}}\).
\end{lemma}
\begin{proof}
  De Morgan algebras have the property of the multiplication and addition
  monoids begin idempotent.
  Thus, as noted earlier, polynomials over De Morgan algebras have the property
  that the degree for any variable in the polynomial is at most 1.
  So, for a set of variables \(X\), there are at most \(2^{\abs{X}}\)
  permutations of exponents for each term with values \(0\) or \(1\).
  Thus, for any De Morgan polynomial \(\theta \in K[X]\), the number of terms in
  \(\theta\) are at most \(2^{\abs{X}}\).
\end{proof}

% \begin{todocomment}
%   \begin{itemize}
%     \item For De Morgan polynomials, the size of the current polynomial is in
%           \(O(2^{\abs{Q}})\) and evaluating \(\theta\) is polynomial in the size of
%           each term (so, worst-case \(O(\abs{Q} 2^{\abs{Q}}\)).
%   \end{itemize}
% \end{todocomment}

Thus, from the \Cref{lemma:strel-spltl}, \Cref{lemma:depth-lim-reach}, \Cref{lemma:delta-poly} and
\Cref{lemma:demorgan-poly}, we can infer the following:

\begin{corollary}
  For a trace \(\sigma\) of length \(\abs{\sigma}\), computing the weight of the
  trace with respect to an alternating automaton has a worst-case complexity of
  \(O(\abs{\sigma} \abs{Q} 2^{\abs{Q}} b^{k})\), where \(b\) and \(k\) are as
  described in \Cref{lemma:depth-lim-reach}.
\end{corollary}

% \begin{corollary}
%   For a trace \(\sigma\) of length \(\abs{\sigma}\), computing the weight of the trace
%   with respect to an alternating automaton has a worst-case complexity of
%   \(O(\abs{\sigma} \abs{Q} 2^{\abs{Q}})\).
% \end{corollary}

\begin{remark}
  While the worst-case complexity of an individual update can be exponential,
  packages like CUDD \cite{somenzi1997cudd} have highly optimized
  implementations of these operations for Boolean algebras.
  In our experiments, we build on these structures to monitor the run of an automaton.
  Moreover, the exponential runtime for enumerating paths for the spatial operators is not encountered in practice for general varieties of spatio-temporal specifications.
\end{remark}

% !TEX root = ../main.tex

\section{Case Study}

\begin{figure*}[htpb]
  \centering
  % First row of images
  \begin{subfigure}[t]{0.25\textwidth}
    \centering
    \includegraphics[width=\linewidth]{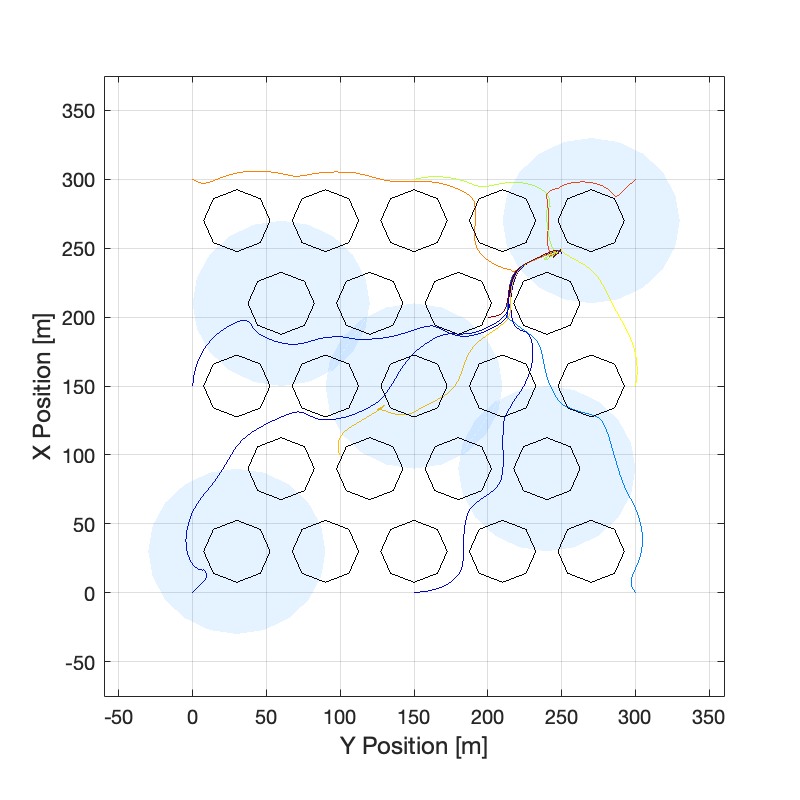}
    \caption{Map 1 trajectories with 10 agents, \\
    23 obstacles and 75\% obstacle density}
  \end{subfigure}
  % \hfill
  \hspace{1em}%
  \begin{subfigure}[t]{0.25\textwidth}
    \centering
    \includegraphics[width=\linewidth]{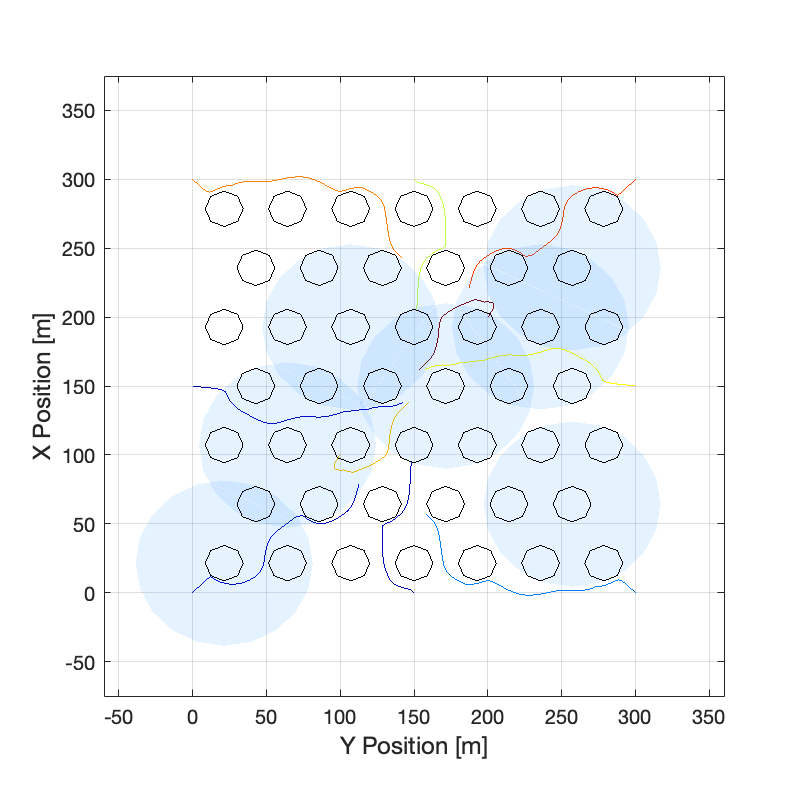}
    \caption{Map 2 trajectories with 10 agents, \\ 46 obstacles and 60\% obstacle density}
  \end{subfigure}
  % \hfill
  \hspace{1em}%
  \begin{subfigure}[t]{0.25\textwidth}
    \centering
    \includegraphics[width=\linewidth]{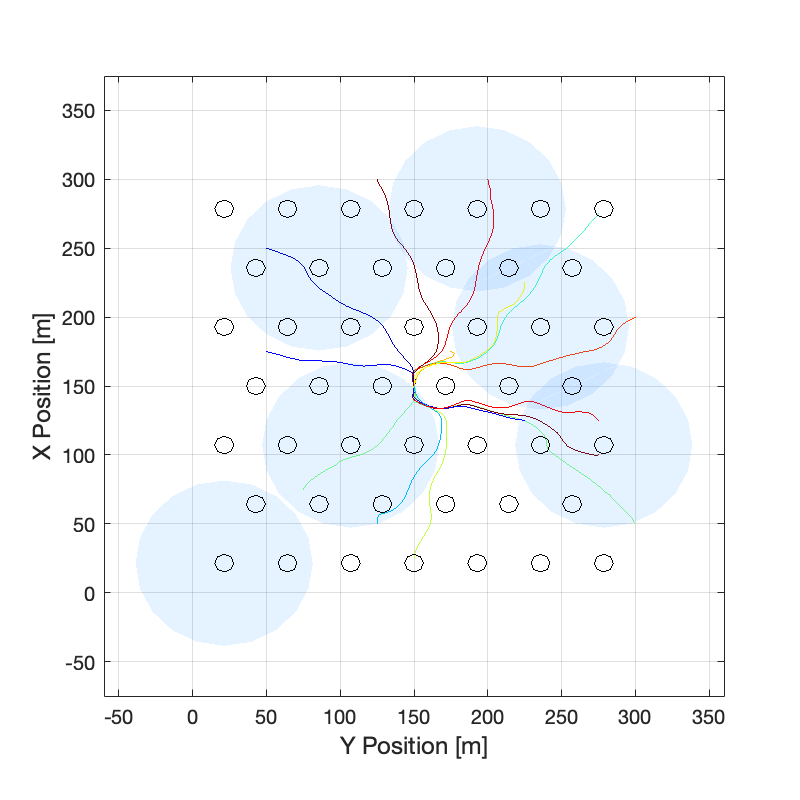}
    \caption{Map 3 trajectories with 15 agents, \\ 46 obstacles and 30\% obstacle density}
  \end{subfigure}
  % \vspace{1em} % Space between rows
  % Second row of images
  \begin{subfigure}[t]{0.25\textwidth}
    \centering
    \includegraphics[width=\linewidth]{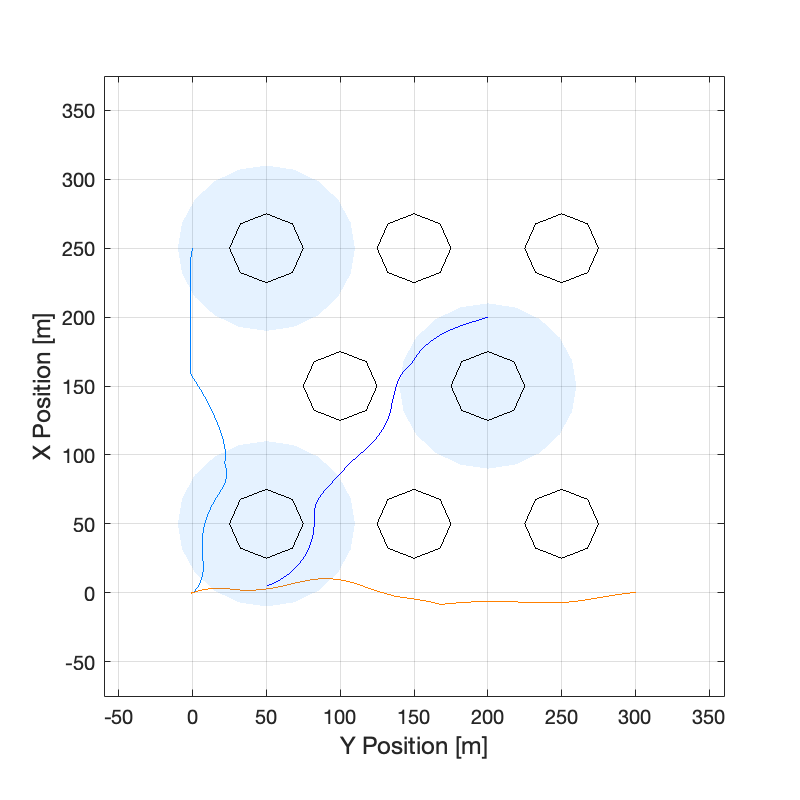}
    \caption{Map 4 trajectories with 3 agents, \\ 46 obstacles and 50\% obstacle density}
  \end{subfigure}
  % \hfill
  \hspace{1em}%
  \begin{subfigure}[t]{0.25\textwidth}
    \centering
    \includegraphics[width=\linewidth]{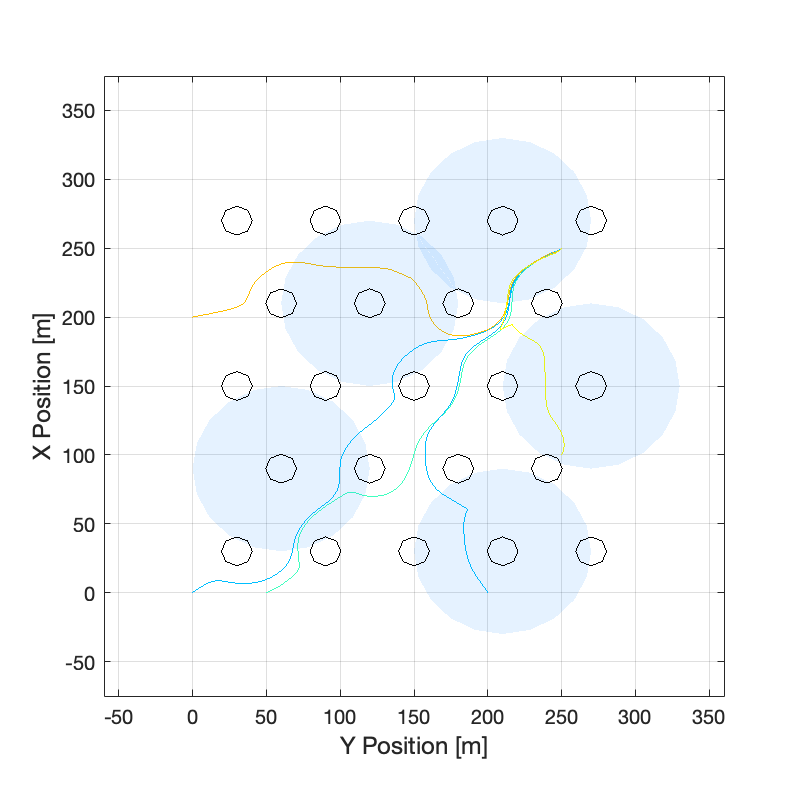}
    \caption{Map 5 trajectories with 5 agents, \\ 23 obstacles and 25\% obstacle density}
  \end{subfigure}
  % \hfill
  \hspace{1em}%
  \begin{subfigure}[t]{0.25\textwidth}
    \centering
    \includegraphics[width=\linewidth]{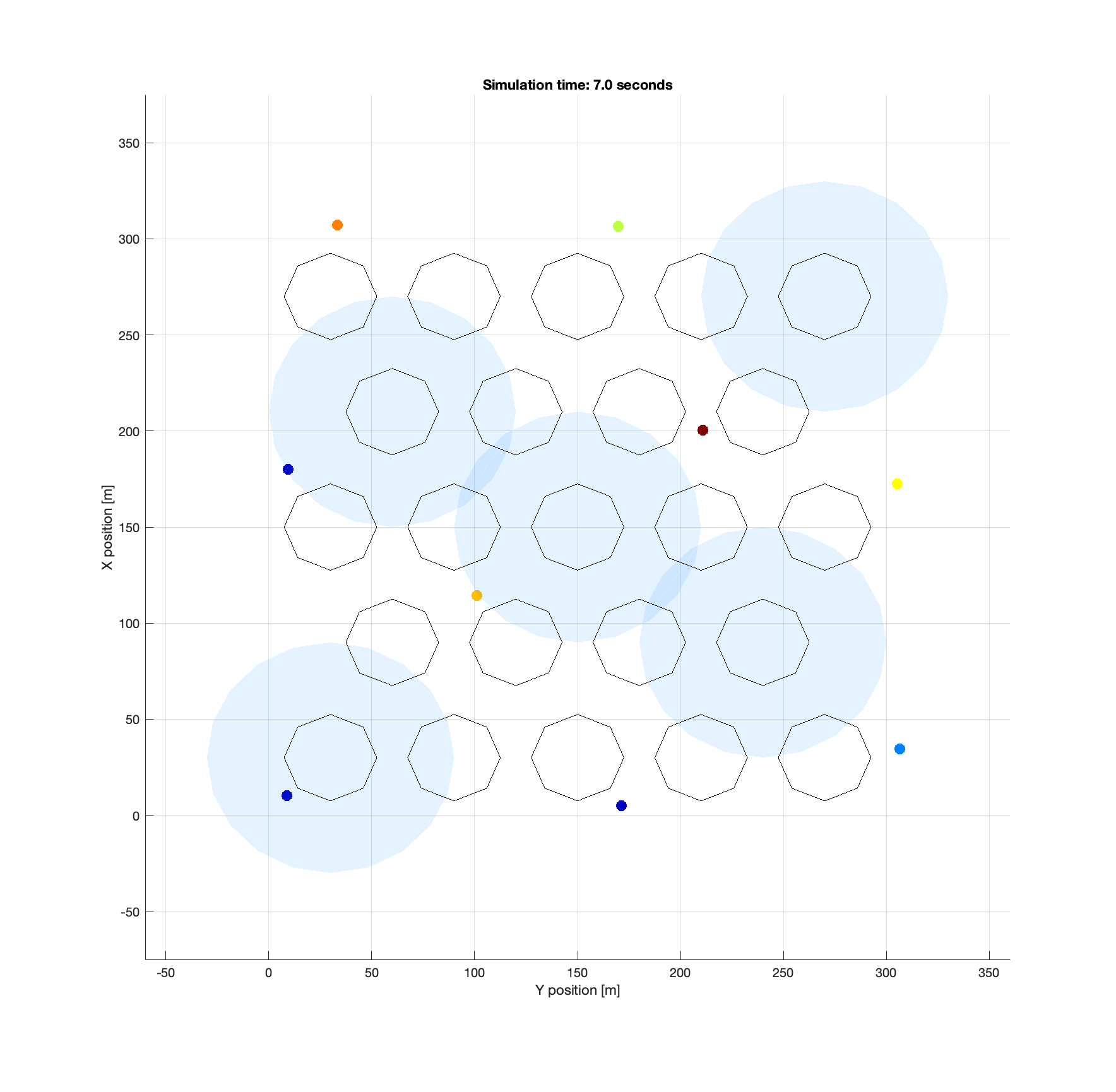}
    \caption{Map 1 with initial positions of drones}
  \end{subfigure}
  \caption{Experimental setup indicating different maps with varying number of obstacles, obstacle densities and number of agents.
    Communication radii of ground stations are indicated by the blue circles.
    \label{fig:maps}}
\end{figure*}

\begin{table*}[htpb]
  \centering
  \begin{tabularx}
    {.85\textwidth}{@{}rrccccYYY@{}}
    \toprule
                           &
                           & \multicolumn{4}{c}{Map configuration}
    % & No. of Drones
    % & No. of Ground stations
    % & No. of Obstacles
    % & Obstacle Density
                           & \multirow[t]{2}{\hsize}{\centering Trace Length~}
                           & \multirow[t]{2}{\hsize}{\centering Seconds per~Trace}
                           & \multirow[t]{2}{\hsize}{\centering Seconds per~Step}
    \\
                           &
                           & \# Drones
                           & \# Stations
                           & \# Obstacles
                           & Obstacle Density
    % & Trace Length
    % & Offline Monitoring
    % & Online Monitoring
    \\
    %map1
    \midrule
    \multirow{2}{*}{Map 1} & $\phi_1$                                              & 10 & 5 & 23 & 0.75 & 6001 & 9.487  & 1.06e-4 \\
                           & $\phi_2$                                              & 10 & 5 & 23 & 0.75 & 6001 & 5.965  & 6.26e-5 \\
    %map2                                                
    \midrule
    \multirow{2}{*}{Map 2} & $\phi_1$                                              & 10 & 7 & 46 & 0.6  & 6001 & 13.231 & 1.31e-4 \\
                           & $\phi_2$                                              & 10 & 7 & 46 & 0.6  & 6001 & 10.761 & 1.05e-4 \\
    %map3                                                
    \midrule
    \multirow{2}{*}{Map 3} & $\phi_1$                                              & 15 & 6 & 46 & 0.3  & 8001 & 18.889 & 1.21e-4 \\
                           & $\phi_2$                                              & 15 & 6 & 46 & 0.3  & 8001 & 49.647 & 2.95e-4 \\
    %map4                                                
    \midrule
    \multirow{2}{*}{Map 4} & $\phi_1$                                              & 3  & 3 & 8  & 0.5  & 6001 & 2.996  & 8.14e-5 \\
                           & $\phi_2$                                              & 3  & 3 & 8  & 0.5  & 6001 & 0.1139 & 3.25e-6 \\
    %map5                                                
    \midrule
    \multirow{2}{*}{Map 5} & $\phi_1$                                              & 5  & 5 & 23 & 0.25 & 7501 & 6.891  & 9.09e-5 \\
                           & $\phi_2$                                              & 5  & 5 & 23 & 0.25 & 7501 & 2.542  & 9.14e-5 \\
    \bottomrule
  \end{tabularx}
  \caption{Results Table}
  \label{tab:results_table}
\end{table*}
To demonstrate our framework, we monitor the execution of a simulation of drone
swarms that are coordinating to flock towards a common goal, while avoiding
obstacles in their path.
Each drone in the swarm can communicate within a certain proximity to other
drones and to stationary ground stations.
This setup allows us to generate diverse set of drone trajectories, with various
obstacle and ground station configurations, while monitoring STREL properties
over their execution.

\paragraph{Experimental Setup}

For our experiments, we use the Swarmlab simulator \cite{soria2020swarmlab} to
generate trajectories for drone swarms under various configurations.
In the simulation, each drone must reach a goal area from some initial location
while avoiding obstacles and passing through regions where they can communicate
with a ground station.
The trajectories are sampled at a \(10ms\) period.

We vary experimental parameters such as obstacle density, number of drones,
and number and positions of ground stations.
The various maps with their configurations are shown in \Cref{fig:maps}.

\paragraph{Spatial Model in Traces}
The space universe \(L\) of the system is the set of drones and ground stations,
each annotated with a unique identifier and position in 3D (although the drones
travel almost along the same altitude, thus the vertical dimension can be
discarded).
An edge is formed between locations \(l, l'\) if they are within communication
reach of each other --- \(40m\) if either of the locations is a ground station,
\(30m\) otherwise.

\paragraph{Specifications}
Here, we look at two specifications over the traces of the spatial models.
These specifications are designed to ensure that certain behavioral conditions
are met by the swarm throughout the simulation.
In particular, we focus on specifications that impose connectivity constraints
on each drone --- such as being in proximity to a ground station or other drones
that are within the proximity of one.

\textbf{Establish Communication \(\phi_1\)}:
Ensures that each drone can maintain connectivity with either the swarm of
drones or a ground station within a specific distance.
If a drone is isolated from the rest of the swarm it must reconnect within 100
time steps (or 1 second) either with another drone that is within 2 hops from a
drone or a ground station.
This specification is intended to monitor the robustness of the swarm's
communication network and ability to dynamically reconfigure in response to
disruptions.
\begin{displaymath}
  \phi_1 = \Alw\left(
  \diamonddiamond_{[1,2]}^{hops} \textsf{drone}
  \lor
  \Ev_{[0, 100]} \diamonddiamond_{[1,2]}^{hops} \left(\textsf{drone} \lor
  \textsf{groundstation}\right)
  \right)
\end{displaymath}

\textbf{Reach Avoid \(\phi_2\)}:
This ensures that drones must avoid obstacles throughout their trajectories and
until it reaches the goal, it must remain in communication with a ground station
within two hops.
This specification is designed to assess navigational safety in obstacle avoidance and connectivity reliability with ground stations as the drones move along their trajectories to their goal positions.
\begin{displaymath}
  \phi_2 =
  \Alw \neg \textsf{obstacle} \land \left(\left( \textsf{drone} \Reach^{hops}_{[0,2]}
    \textsf{groundstation}\right) \Until \textsf{goal}\right)
\end{displaymath}

The predicates \(\textsf{drone}\) and \(\textsf{groundstation}\) are enabled at
a location if the location corresponds to a drone or ground station respectively.

In the specifications, we refer to the distance function \(hops\) that labels
each edge in a spatial model with the weight \(1 \in \Ne^\infty\), and the
distance domain is the counting domain, i.e., \((\Ne^\infty, 0, \infty, +,
\leq)\).
This distance function applied over a path in the spatial model essentially
captures the number of ``hops'' some message would make for an agent at the
start location to communicate to the agent at the other end of the path.

\paragraph{Reporting results}
For each map configuration and specification, we perform both, online monitoring
and offline monitoring.
For the former, we report the average time taken for each step across all drones
as ego locations.
In the latter, we report the average total time taken for monitoring an entire
trajectory for all drones.
To mitigate timing outliers due to operating system background processes, we ran
10 monitoring trials across 7 batches and compute the average time across the 7
best trials.
See \Cref{tab:results_table} for the details.

\subsection{Results and Analysis}
Our experiments were conducted on a laptop with an M3 chip featuring an 8-core
CPU, 10-core GPU and a 16-core Neural engine.
The framework operates at a sampling rate of 10 milliseconds per step which
allows us to capture fine-grained behaviors.
For each configuration we observe that the mointoring time is influenced by the
obstacle density, number of drones and trace length.
As expected, configurations with fewer drones and obstacles such as Maps 4 and 5
have significantly lower monitoring times.
Configurations with higher obstacle densities and larger number of drones (Map
3) result in increased monitoring times, especially for $\phi_2$.
However, even despite the theoretical complexities of monitoring over multiple
agents and a high number of obstacles, the actual computation times per step are
not affected, as seen in \autoref{tab:results_table} in the seconds-per-step
metric.
This indicates that our framework can handle a large number of spatial
locations, agents and map configurations without significant slowdown.
% \begin{todocomment}
%   \begin{itemize}
%     \item Use Case Study from older papers
%     \item Use Nick's example
%   \end{itemize}
% \end{todocomment}

% !TEX root = ../main.tex

\section{Related Work}

In the work by \cite{borzoo1}, the authors propose distributed monitoring for an STL formula over the global state of a distributed system.
For a truly distributed system, estimating the global state is a hard problem due to asynchronous communication and lack of a central clock.
The authors show that STL formulas can be monitored under the assumption of bounded clock drift using SMT solvers.
In our approach, we do not consider issues arising from clock asynchrony, and assume a central clock.
It would be interesting to extend our approach with assumptions similar to the above work.

The work by \cite{matos} introduces a method for monitoring continuous-time, continuous-valued signals in CPS by compensating for clock drifts, enabling fast, distributing monitoring suitable for real-time deployment in autonomous systems, while the authors of \cite{momtaz2023predicate} work on distributed predicate detection using the notion of consistent cuts extended to continuous-time settings.
STL has also been extended to applications in monitoring agent-based epidemic
spread models \cite{bortolussi2014specifying}.
Similar to our work, \cite{audrito} focus on integration spatial and temporal logic in distributed monitoring of swarms.
The work \cite{li2021runtime} presents STSL, a spatio-temporal specification language combining temporal and spatial logic $\mathcal{S}4_u$ for CPS, with applications in runtime verification, falsification, and parameter synthesis.
TQTL has been introduced in \cite{anand1} to assess spatio-temporal properties of perception algorithms in autonomous systems allowing quality evaluation even without ground truth labels.
Spatio-Temporal Perception Logic (SPTL) \cite{hekmatnejad2024formalizing,balakrishnan2021percemona} has also been introduced to evaluate perception systems in autonomous vehicles providing spatial-temporal reasoning capabilities and enabling offline performance checks without ground-truth data.
Recently, team Spatio-Temporal Reach and Escape Logic (tSTREL) \cite{liu2025tstrel} has been introduced as an extension of STREL with the possibility to evaluate reach and escape operators over a set of locations, which provides more flexibility in expressing multi-agent system properties.

% !TEX root = ../main.tex

\section{Discussions}
\subsection{Limitations}

In this paper, we currently only handle monitoring for finite-length traces, where the input word is some \(\sigma \in \Sigma^*\).
To support infinite-length traces, we would require the alternating word automaton to have Buchi acceptance conditions \cite{vardi1995alternating}.
In general, while our construction works for translating STREL specifications into such automata, we will explore the semantics and applications for them in future work.

Furthermore, the semantics for STREL in \cite{bartocci2017monitoring} and \cite{nenzi2022logic} are defined for continuous signals from multiple locations.
Handling this would require extending our construction to \emph{alternating timed automata}.

In our construction, we also consider only spatial models with finite, a priori known locations and do not handle edge multiplicities, and will investigate efficient means to do so in the future.
For example, for the former, we can potentially use a symbolic variable for the locations and realize them as we see new locations in the input word.
\subsection{Conclusion}
In this work, we address the growing need for the robust monitoring of spatially distributed cyber-physical systems by leveraging STREL. We introduce a novel, linear-time construction of alternating finite automata (AFA) from STREL specifications. Additionally we presented an extension to weighted AFA for quantitative monitoring. Our contributions allow for efficient and scalable, offline and online monitoring as demonstrated by our Swarmlab experiments. By analyzing drone trajectories in various map configurations, we validated our framework's capability to enforce reach-avoid, and connectivity-maintenance specifications on drones that attempt to navigate their way to the goal in complex environments. To summarize, we highlight the utility of automaton-based monitoring for scalable and dynamic applications of CPS and its efficiency in our experiments indicates promising applications in real-world distributed systems.

\begin{acks}
  This work was partially supported by the \grantsponsor{nsf}{National Science
    Foundation}{https://www.nsf.gov/} through the following grants:
  \grantnum{nsf}{CAREER award (SHF-2048094)},
  \grantnum{nsf}{CNS-1932620},
  \grantnum{nsf}{CNS-2039087},
  \grantnum{nsf}{FMitF-1837131},
  \grantnum{nsf}{CCF-SHF-1932620},
  \grantnum{nsf}{IIS-SLES-2417075};
  funding by
  \grantsponsor{toyota}{Toyota R\&D}{https://amrd.toyota.com/} and
  \grantsponsor{siemens}{Siemens Corporate Research}{https://www.siemens.com/global/en/company/innovation/research-development.html}
  through the USC Center for Autonomy and AI;
  a grant from
  \grantsponsor{fort}{Ford Motors}{https://corporate.ford.com/home.html};
  an \grantsponsor{amazon}{Amazon}{https://www.amazon.science} \grantnum{amazon}{Faculty Research Award};
  and the
  \grantsponsor{airbus}{Airbus Institute for Engineering Research}{https://viterbischool.usc.edu/aier/}.
  This study was carried out within the PNRR research activities of the
  consortium iNEST (Interconnected North-Est Innovation Ecosystem) funded by the
  \grantsponsor{eu}{European Union Next-GenerationEU}{https://next-generation-eu.europa.eu}
  (\grantnum{eu}{PNRR – Missione 4 Componente 2, Investimento 1.5 – D.D.
    1058 23/06/2022, ECS\_00000043})
  and
  the \grantnum{eu}{PRIN project 20228FT78M DREAM}
  (modular software design to reduce uncertainty in ethics-based cyber-physical systems).
  This work does not reflect the views or positions of any organization listed.
\end{acks}

% \newpage
\bibliographystyle{ACM-Reference-Format}
\bibliography{bib}

\end{document}